\documentclass[10pt,twoside,a4paper]{amsart}
\leftmargin - 1.2cm
\usepackage[centertags]{amsmath}
\usepackage{amsfonts}
\usepackage{amssymb}
\usepackage{amsthm}
\usepackage{newlfont}
\usepackage{epsfig}
\usepackage{amscd}

\newcommand{\TT}{{\mathbb T}}
\newcommand{\HH}{{\mathbb H}}
\newcommand{\ZZ}{{\mathbb Z}}
\newcommand{\KK}{{\mathbb K}}

\newcommand{\VV}{{\mathbb V}}

\newcommand{\bv}{{\boldsymbol v}}

\theoremstyle{plain}
\newtheorem{Th}{Theorem}[section]
\newtheorem{Cor}[Th]{Corollary}

\newtheorem{Prop}[Th]{Proposition}
\theoremstyle{definition}

\theoremstyle{remark}
\newtheorem{Rem}[Th]{Remark}
\numberwithin{equation}{section}
\begin{document}
\title[Integrable linear problems on the triangular and honeycomb lattices]
{Integrable lattices and their sublattices II.\\
From the B-quadrilateral lattice to the self-adjoint schemes\\ on the triangular 
and the honeycomb lattices
}
\author{A. Doliwa, M. Nieszporski, and P. M. Santini}
\address{A. Doliwa: Uniwersytet Warmi\'{n}sko-Mazurski w Olsztynie,
Wydzia{\l} Matematyki i Informatyki,
ul.~\.{Z}o{\l}nierska 14 A, 10-561 Olsztyn, Poland}
\email{doliwa@matman.uwm.edu.pl}
\address{M. Nieszporski: Katedra Metod Matematycznych Fizyki,
Uniwersytet Warszawski
ul. Ho\.za 74, 00-682 Warszawa, Poland, and
School of Mathematics, University of Leeds, LS2 9JT Leeds, UK}
\email{maciejun@fuw.edu.pl}
\address{P. M. Santini: Dipartimento di Fisica, Universit\`a di Roma
``La Sapienza'' and
Istituto Nazionale di Fisica Nucleare, Sezione di Roma,
Piazz.le Aldo Moro 2, I--00185 Roma, Italy}
\email{paolo.santini@roma1.infn.it}
\date{}

\begin{abstract}
An integrable self-adjoint 7-point scheme on the triangular
lattice and an integrable self-adjoint scheme on the honeycomb lattice are studied
using the sublattice approach. The star-triangle relation between these systems
is introduced, and the Darboux transformations
for both linear problems from the Moutard transformation of the 
B-(Moutard) quadrilateral lattice are obtained. A geometric
interpretation of the Laplace transformations of the self-adjoint 7-point
scheme is given and the corresponding novel integrable discrete 3D system is
constructed.
\end{abstract}
\maketitle

\section{Introduction}
\subsection{Integrable linear systems on triangular and honeycomb grids}
In this paper we study integrable linear problems on planar graphs built of
regular polygons (of the same type) different from squares: the self-adjoint 
linear problem on the star of the regular triangular lattice,
and the self-adjoint linear problem on the honeycomb lattice. 
We call a linear problem integrable, if it possesses a certain number of
relevant mathematical properties, including: the existence of i) discrete
symmetries (Laplace and Darboux type transformations), ii) continuous
symmetries (nonlinear evolutions), iii) dressing procedures
($\bar\partial$-dressing, finite gap constructions) enabling one to construct
large classes of analytic solutions.

To understand the 
integrability of systems on the triangular and honeycomb lattices 
is a challenging subject. Exactly solvable models of statistical mechanics on 
such
grids are rather well studied \cite{Baxter}, but there are not so many papers
discussing this problem from the  point of view of
integrable difference equations (see, however,
\cite{Nijhoff,NovDyn,Adl,AdlSu,AdlVes,BobHofSu,BobHof,BobAg,Schiff,NSD}). 

The self-adjoint linear problem on the triangular lattice 
was introduced into the theory of integrability 
in \cite{Nov,NovDyn} as a discrete analog of the two dimensional 
elliptic Schr\"{o}dinger operator. This discretization was distinguished, among
other operators having the same continuous limit, by the existence of a
decomposition into a sum of a multiplication operator and a factorizable one. 
Such decomposition, or "extended factorization", leads to transformations,
called in the literature Laplace transformations \cite{Nov,NovDyn}. The
transformations preserve the form of the operator but does not
involve any parameters.
A more general transformation of the Darboux type (with functional parameters)
for the self-adjoint linear
problem on the triangular lattice was 
introduced and studied in \cite{NSD}. 

Integrable circle patterns on the plane with the honeycomb combinatorics were
studied in \cite{BobHofSu,BobHof,BobAg} not only from the point of view of
integrable systems, but also as potentially important objects of
the discrete holomorphic function theory. In \cite{DynNov} special Laplace
transformation operators
on the triangular lattice were studied, within that context, as analogs of the
$\partial$ and $\bar\partial$ operators.
In \cite{BobHofSu} circle packings
exhibiting the so called multi-ratio property were investigated, while
the patterns with constant intersection
angles were studied in \cite{BobHof}. A special
reduction of the packing with constant angles leading to the discrete
Painlev\'{e} and Riccati equations was the subject of \cite{BobAg}.

\subsection{Main ideas and results}

The theory of integrable difference (or discrete) equations on the $\ZZ^N$ grids
is much more developed. 
It turns out that many
integrable $\ZZ^N$ systems can be obtained by reductions of the 
multidimensional quadrilateral lattice \cite{MQL,DS-EMP,DS-sym,DMS} 
(multidimensional lattices with all elementary quadrilaterals planar), called
also discrete conjugate net~\cite{Sauer}. One of the
goals of 
this paper is to show that the integrable systems on
triangular and honeycomb grids can be incorporated into the multidimensional
quadrilateral lattice theory.

By imposing an additional linear constraint on the 
quadrilateral lattice, one obtains the so called B-quadrilateral lattice 
\cite{BQL},
which is an important object in the paper. On the algebraic
level, such a lattice is described by the linear system of discrete
Moutard equations \cite{DJM-dBKP,NimmoSchief}, the compatibility condition 
of which gives the 
nonlinear Miwa equations  \cite{Miwa}, called also the discrete BKP equations
(discrete Kadomtsev--Petviashvili equations of type B). Recently, 
integrability properties of the two-dimensional version of this lattice were
used to explain integrability of the self-adjoint 5-point scheme on the star 
of the $\ZZ^2$ lattice \cite{DGNS}. The approach presented there was based on
the idea of deriving from integrable equations on a lattice (and from their
integrability features) novel integrable equations on a sublattice of the
original lattice, together with their integrability features.  
In particular, i) the Darboux transformation for the self-adjoint 
5-point scheme
given in \cite{NSD} was rederived in \cite{DGNS} from the Moutard transformation
\cite{NimmoSchief} of the discrete Moutard equation, ii) the finite-gap theory
for the self-adjoint 5-point scheme was obtained as direct consequence of the
finite gap theory (constructed also in \cite{DGNS}) for the discrete Moutard
equations in $\ZZ^N$.

\begin{figure}
\begin{center}
\includegraphics[width=6cm]{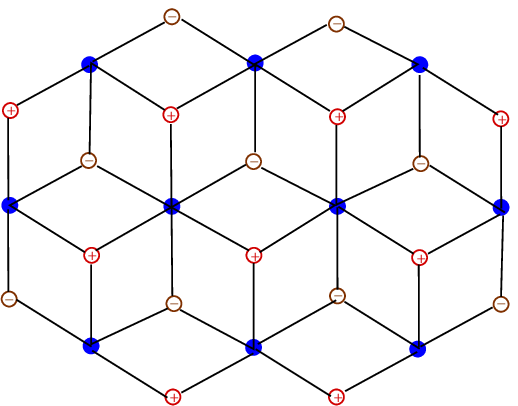} \hskip 1cm
\includegraphics[width=6cm]{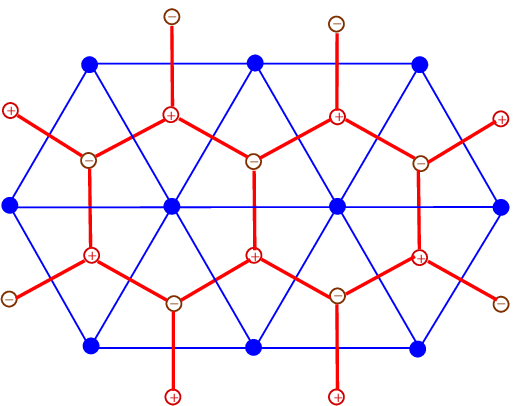}
\end{center}
\caption{The diagonal staircase section as the bipartite quasiregular rhombic
tiling (on the left), and the corresponding triangular and honeycomb lattices
(on the right)}
\label{fig:T-H-lattices}
\end{figure}  
This is also the main idea behind the present paper. We derive the integrability
properties of the self-adjoint linear problems on the triangular and honeycomb
lattices from the corresponding properties of the three dimensional
B-quadrilateral lattice. In particular, we obtain in this way the Darboux-type
transformations for both linear problems, and we show that the Laplace
transformations of \cite{Nov} are nothing but a relation between two subsequent
triangular sublattices of the B-quadrilateral lattice.
We remark that the 
approach we use here is very close, in spirit, to the one used in 
\cite{BobHof,BobAg}, 
where properties of the hexagonal
circle patterns with constant intersection
angles  were derived from properties of an 
integrable system on the $\ZZ^3$ lattice, and to that used in \cite{BobMerSur},
where discrete holomorphic functions on quad-graphs were constructed from
discrete holomorphic functions on $\ZZ^N$. 
Our results can be 
also interpreted from the
point of view of the discrete
Cauchy--Riemann (discrete Moutard) equation on quad-graphs \cite{BobSur1,Mercat}, 
with the graph being
the quasiregular rhombic tiling, as visualized on Figure \ref{fig:T-H-lattices}. 
However, we would like to stress the multidimensional 
origin of our constructions and their geometric meaning.

The layout of the paper is as follows.
In Section \ref{sec:T-H-lattices}, after construction of the triangular and
honeycomb grids as sublattices of the diagonal staircase section of the $\ZZ^3$
lattice (see Figure \ref{fig:T-H-lattices}),
we derive the self-adjoint 7-point scheme and the honeycomb linear
problem from the discrete Moutard system. The
sublattice approach provides the clear geometric meaning to the
star-triangle relation between both linear
problems. In Section \ref{sec:Laplace} we present
the corresponding interpretation of the Laplace transformation of the
self-adjoint 7-point scheme \cite{Nov} as
the transition between subsequent triangular sublattices of the B-quadrilateral
lattice. We introduce a 3D fully discrete novel integrable system describing such
a transition. Finally, in Section \ref{sec:DT} 
we obtain the Darboux transformations of the self-adjoint 7-point scheme 
\cite{NSD} and the honeycomb linear problem
from the Moutard transformation.
In the rest of this introductory section we recall necessary informations on the
discrete BKP equation, its linear problem (the discrete Moutard system), and
the corresponding Moutard transformation.

\subsection{The discrete Moutard system (the B-quadrilateral lattice)}
Consider the system of discrete Moutard equations (the discrete BKP linear
system)~\cite{DJM-dBKP,NimmoSchief}
\begin{equation} \label{eq:Moutard-ij}
\psi_{ij} - \psi = f^{ij} \times (\psi_{i} - \psi_{j}), \qquad
1\leq i < j \leq N,
\end{equation}
where $\psi:\ZZ^N\to\VV$, and $\VV$ is a linear space over a field 
$\mathbb{K}$, and $f^{ij}:\ZZ^N\to\KK$, $1\leq i < j \leq N$, are some 
functions; 
here and in all the paper
by a subscript we denote the shift in the corresponding discrete variable, i.e.
$f_i(n_1, \dots , n_i, \dots, n_N) = f(n_1, \dots , n_i+1, \dots, n_N)$.

\begin{Rem}
In \cite{BQL} it was shown that the discrete Moutard system 
\eqref{eq:Moutard-ij} characterizes algebraically (up to a gauge transformation)
special quadrilateral lattices 
$x = [\psi]:\ZZ^N \to \mathbb{P}(\VV)$ in the projectivization of
 $\VV$, which satisfy the following additional local 
linear constraint:
any point $x$ of the lattice and its neighbours
$x_{ij}$, $x_{ik}$ and $x_{jk}$, for $i,j,k$ distinct, are coplanar. 
\end{Rem}
The coefficients $f^{ij}$ are not arbitrary functions.
Compatibility of the Moutard system
leads to the following set of nonlinear equations \cite{NimmoSchief}
\begin{equation} \label{eq:nonl-BQL-f}
1 + f^{jk}_{i}(f^{ij} - f^{ik}) = f^{ik}_{j} f^{ij}= f^{ij}_{k}f^{ik}, 
\qquad i,j,k \quad \text{distinct},
\end{equation}
where, formally, we put $f^{ji}=-f^{ij}$.

If $\theta:\ZZ^N\to\KK$ is a scalar
solution of the system \eqref{eq:Moutard-ij} then the solution
$\tilde\psi:\ZZ^N\to\VV$ of the Moutard transformation 
equations \cite{NimmoSchief}
\begin{equation} \label{eq:Mt}
\tilde\psi_{i} - \psi = \frac{\theta}{\theta_{i}}(\tilde\psi - \psi_{i}),
\qquad 1\leq i\leq N,
\end{equation}
satisfies a new discrete Moutard system with the fields
\begin{equation}
\tilde{f}^{ij} = f^{ij}\frac{\theta_{i}\theta_{j}}{\theta_{ij}\theta}. 
\end{equation} 
We will need also the fact that $\theta^{-1}$ satisfies the same Moutard system
as the tranformed wave function $\tilde\psi$.
\begin{Rem} \label{rem:B-ft}
The discrete Moutard transformation provides the algebraic part of the
B-reduction of the fundamental transformation of the quadrilateral lattice
\cite{BQL}. This transformation acts within the class of B-quadrilateral
lattices and, apart from the usual property of the fundamental transformation
stating that the corresponding
points $x$, $x_{i}$, $\tilde{x}$, $\tilde{x}_{i}$ of the two
quadrilateral lattices are coplanar, we have the following
additional linear constraint:
the points $x$, $x_{ij}$,
$\tilde{x}_{i}$ and $\tilde{x}_{j}$, for $i \ne j$, are coplanar as well. 
\end{Rem}

\section{Linear problems on the triangular and honeycomb lattices}
\label{sec:T-H-lattices}

The goal of this section is to present the derivation of the discrete Laplace
equations on the triangular and honeycomb lattices from the system of Moutard
equations on $\ZZ^3$. This is done in full detail to prepare the setting 
for the new results in 
Sections \ref{sec:Laplace} and \ref{sec:DT}. The derivation of
the Laplace equations from the discrete Moutard (Cauchy--Riemann) equations on
planar bipartite quad-graphs is well known \cite{Mercat,BobSur1}. The $\ZZ^3$
origin of some integrable systems on planar graphs was used in
\cite{BobHof,BobAg}, 
where properties of the hexagonal
circle patterns with constant intersection
angles  were derived from properties of an 
integrable system on the $\ZZ^3$ lattice, and in \cite{BobMerSur},
where discrete holomorphic functions on quad-graphs were constructed from
discrete holomorphic functions on $\ZZ^N$.   

\subsection{The staircase section}
Define the 
bipartition of the $\ZZ^3$ lattice depending if the sum of the
coordinates is even (black points) or odd (white points). Consider the subset 
\begin{equation*}
V(\TT) = \{ m (1,1,0) + n (0,1,1): \; m,n\in\ZZ\} = 
\{ (n_1,n_2,n_2)\in\ZZ^3: \; n_1-n_2+n_3 = 0 \},
\end{equation*} 
of the black point lattice, which can be considered as the
intersection of the $\ZZ^3$ lattice with the diagonal plane.
Intersecting, instead, the $\ZZ^3$ lattice with two "nearests" 
parallel planes 
we obtain 
the subsets
\begin{align*}
V(\TT_+) & = \{ (1,0,0) + m (1,1,0) + n (0,1,1): \; m,n\in\ZZ\} =
\{ (n_1,n_2,n_2)\in\ZZ^3: \; n_1-n_2+n_3 - 1 =0\},\\
V(\TT_-) & = \{ (0,1,0) + m (1,1,0) + n (0,1,1): \; m,n\in\ZZ\} =
\{ (n_1,n_2,n_2)\in\ZZ^3: \; n_1-n_2+n_3 + 1 =0 \},
\end{align*}
of the white point lattice. We call the union of the above sets, 
together with the
corresponding edges and facets of the $\ZZ^3$ lattice, the 
diagonal staircase section of $\ZZ^3$ (see Figure \ref{fig:T-H-lattices}).
As a planar graph the diagonal staircase section is often refered to as
quasiregular rhombic tiling (the dual of
the kagome lattice).

Points of $V(\TT)$ are vertices of the triangular grid $\TT$ whose set of
edges $E(\TT)$ consists of segments connecting the nearest neighbours 
\begin{equation*}
E(\TT) = \{ [\boldsymbol{v}_1,\boldsymbol{v}_2]: \; 
\boldsymbol{v}_1,\boldsymbol{v}_2 \in V(\TT), \; 
|\boldsymbol{v}_1 - \boldsymbol{v}_2 | = \sqrt{2} \}.
\end{equation*}
Similarly we define the triangular grids $\TT_\pm$, which will be used in Section
\ref{sec:Laplace}. However we will need first the honeycomb grid $\HH$ whose set
of vertices is $V(\HH) = V(\TT_+)\cup V(\TT_-)$, and edges of which
are the segments connecting points of $V(\TT_\pm)$ with their
nearest $V(\TT_\mp)$ neighbours, see Figure \ref{fig:T-H-lattices}
\begin{equation*}
E(\HH) = \{ [\boldsymbol{v}_+,\boldsymbol{v}_-]: \; 
\boldsymbol{v}_+\in V(\TT_+), \boldsymbol{v}_- \in V(\TT_-), \; 
|\boldsymbol{v}_+ - \boldsymbol{v}_- | = \sqrt{2} \}.
\end{equation*}
We keep the bipartition of the honeycomb lattice into $\pm$-points.
The honeycomb tiling is the dual of the triangular tiling (in the
standard crystallographic sense). Roughly speaking, 
this means that the centers of the facets of one
lattice are the vertices of the second one (see Figure~\ref{fig:T-H-lattices}). 

\subsection{The self-adjont 7 point scheme and the sublattice approach}
\label{sec:7-point}
The goal of this Section is to derive, using the sublattice approach, 
the self-adjoint 7-point scheme from the system of the Moutard equations on 
$\ZZ^3$ grid (the three dimensional B-quadrilateral lattice). However, we first
recall a way (see, for example \cite{BobMerSur}) in which one can construct
the Laplace equation for an arbitrary graph $\mathcal{G}$ with given weight 
functions $\nu:E(\mathcal{G})\to\KK$ from non oriented edges to the field $\KK$. 
When $\psi:V(\mathcal{G})\to\VV$ is a
function on vertices of the graph, then it satisfies
the corresponding Laplace equation if at each vertex $\bv_0$
\begin{equation} \label{eq:gen-Laplace}
\sum_{\bv\sim\bv_0} \nu(\bv,\bv_0) (\psi(\bv) -\psi(\bv_0)) =0,
\end{equation}
where the summation is over vertices adjacent to $\bv_0$.

In the case of the triangular grid there are
three families of edges, and therefore function $\nu$ gives rise to three
functions, say
$a,b$ and $s$, given on triangular lattice; see Figure \ref{fig:T-lattice}, 
where also the corresponding
convention of labeling directions $I$ and $I\!I$ of the lattice is given.
\begin{figure}
\begin{center}
\includegraphics[width=5cm]{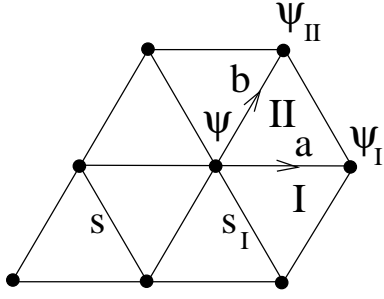}
\end{center}
\caption{The triangular lattice, its main directions and the fields $a,b,s$ on
its edges}
\label{fig:T-lattice}
\end{figure}  
The Laplace equation on such weighted graph reads
\begin{equation} \label{eq:s-a-a-7-p} 
a(\psi_{I} - \psi) + a_{-I}(\psi_{-I} - \psi) +
b(\psi_{II} - \psi) + b_{-I\!I}(\psi_{-I\!I} - \psi) +
s_{I}(\psi_{I,-I\!I} - \psi) + s_{II}(\psi_{-I,I\!I} - \psi) =0,
\end{equation}
where $\psi:V(\TT)\to\VV$ is a field defined in the vertices of the
grid with values in a vector space $\VV$. It is easy to see that
given any solution $\psi:\ZZ^3\to\VV$
of the the Moutard linear system \eqref{eq:Moutard-ij}, then its restriction 
to the vertices $V(\TT)$ of the triangular 
lattice satisfies the self-adjoint affine 7-point scheme \eqref{eq:s-a-a-7-p}
with the edge fields defined as
\begin{equation} \label{eq:subl-notation}
a = \frac{1}{f^{12}}, \qquad b = \frac{1}{f^{23}}, \qquad s = -
f^{13}_{-1-2-3}, 
\end{equation}
and the shifts on the triangular lattice given by 
\begin{equation} \label{eq:TL-shifts}
(I) = (12), \qquad \text{and} \qquad
(II) = (23).
\end{equation}

Equation \eqref{eq:s-a-a-7-p} is the affine form of the general self-adjoint
scheme on the triangular grid
\begin{equation} \label{eq:gen-7p-sa}
A\Psi_{I} + A_{-I}\Psi_{-I} +
B\Psi_{II} + B_{-I\!I}\Psi_{-I\!I} +
S_{I}\Psi_{I,-I\!I}+ S_{II}\Psi_{-I,I\!I} = F \Psi,
\end{equation}
considered in \cite{Nov}. In this paper we confine ourselves to consideration of
the affine scheme \eqref{eq:s-a-a-7-p} (without loss of generality because, 
as we will
see below, every scheme \eqref{eq:gen-7p-sa} can be transformed to an affine 
form \eqref{eq:s-a-a-7-p} by means of gauge transformation).

\begin{Cor} \label{cor:gen-aff}
Given a scalar solution $\sigma$ of the generic self-adjoint 7-point scheme
\eqref{eq:gen-7p-sa},
then $\psi = \Psi/\sigma$ satisfies the affine self-adjoint 7-point scheme
\eqref{eq:s-a-a-7-p} with the coefficients
\begin{equation} \label{eq:coeff-gen-aff}
a = A\sigma_{I}\sigma, \qquad b = B\sigma_{I\!I}\sigma, \qquad
s = S \sigma_{-I}\sigma_{-I\!I}.
\end{equation}
Conversely, given a generic function $\rho:V(\TT)\to\KK$ on the triangular 
grid and given a solution $\psi:V(\TT)\to\VV$ 
of the self-adjoint 7-point scheme \eqref{eq:s-a-a-7-p}, then the function 
$\Psi=\psi/\rho$ satisfies the scheme \eqref{eq:gen-7p-sa}
with the coefficients
\begin{equation} 
A = a\rho_{I}\rho, \qquad B = b\rho_{I\!I}\rho, \qquad
S = s \rho_{-I}\rho_{-I\!I}, \qquad 
F = (a + a_{-I} + b + b_{-I\!I} + s_{I,-I\!I} + s_{-I,I\!I})\rho^2.
\end{equation}
\end{Cor}

\begin{Rem}
Since the Moutard equation on the $\ZZ^N$ grid can be integrated by the finite
gap method \cite{DGNS}, the above sublattice derivation of the 7-point scheme
from the discrete Moutard equations imply that also the 7-point scheme on the
triangular lattice can be integrated using the theory developed in
\cite{DGNS,BQL}.
\end{Rem}
\begin{Rem} \label{rem:Paolo-gen}
Notice that the above derivation of the self-adjoint affine 7-point scheme can
be generalized to arbitrary dimension as follows. Consider a solution 
$\psi:\ZZ^N\to\VV$
of the the $N$-dimensional 
Moutard linear system \eqref{eq:Moutard-ij} on the $2N$ elementary
quadrilaterals
\begin{equation*}
(1,2), \dots , (j,j+1), \dots , (N-1,N), (N,-1), (-1,-2), \dots , (-N,1).
\end{equation*}
Add the first $N$ schemes and substract the remaining ones (dividing by the
Moutard coefficients, when necessary). The result is the following self-adjoint
$(2N+1)$-point scheme (on a star with $2N$ legs) on the black points:
\begin{equation}
\sum_{j=1}^{N-1} \left( \frac{\psi_{j,j+1} - \psi}{f^{j,j+1}} + 
\frac{\psi_{-j,-j-1} - \psi}{f^{j,j+1}_{-j,-j-1}} \right) =
f^{1,N}_{-1}( \psi_{-1,N} - \psi) + f^{1,N}_{-N}( \psi_{1,-N} - \psi).
\end{equation}
\end{Rem}

\subsection{The linear problem on the honeycomb lattice}
\label{sec:honeycomb}
The discrete Moutard equations
\eqref{eq:Moutard-ij}  imply a linear relation between a point of the $\TT_-$
lattice with its three nearest $\TT_+$ neighbours: 
\begin{equation}
f^{12}(\psi_{1} - \psi_{2}) + f^{23}(\psi_{3} - \psi_{2}) -
\frac{1}{f^{13}_{2}}(\psi_{123} - \psi_{2})  = 0, 
\end{equation}
and a 
linear relation between a point of the $\TT_+$
lattice with its three nearest $\TT_-$ neighbours: 
\begin{equation}
f^{12}(\psi_{2} - \psi_{1}) + f^{23}_{1-3}(\psi_{12-3} - \psi_{1}) -
\frac{1}{f^{13}_{-3}}(\psi_{-3} - \psi_{1})  = 0. 
\end{equation}
These two relations, which give the geometric characterization of the
B-quadrilateral lattice \cite{BQL}, can be interpreted as the
following linear self-adjoint system of equations on the honeycomb white lattice
(see Figure \ref{fig:H-lattice})
\begin{figure}
\begin{center}
\includegraphics[width=6cm]{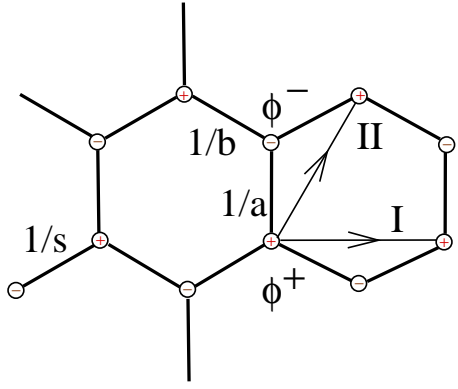}
\end{center}
\caption{The honeycomb lattice}
\label{fig:H-lattice}
\end{figure}
\begin{align} \label{eq:hcl1}
\frac{1}{a}(\phi^+ - \phi^-) + \frac{1}{b}(\phi^+_{-I,I\!I} - \phi^-) +
\frac{1}{s_{I,I\!I}}(\phi^+_{I\!I} - \phi^-) & = 0,\\
\label{eq:hcl2}
\frac{1}{a}(\phi^- - \phi^+) + \frac{1}{b_{I,-I\!I}}(\phi^-_{I,-I\!I} - \phi^+) +
\frac{1}{s_{I}}(\phi^-_{-I\!I} - \phi^+) & = 0.
\end{align}
for the fields $\phi^+$ and $\phi^-$, defined by
\begin{equation} \label{eq:phi-pm}
\phi^+=\psi_{1}, \qquad \phi^-=\psi_{2}, 
\end{equation}
having used also \eqref{eq:subl-notation} and the triangular lattice shifts
\eqref{eq:TL-shifts}.

Again, like in the triangular lattice, one can consider a more 
general linear system
on the honeycomb lattice than the above affine one
\eqref{eq:hcl1}-\eqref{eq:hcl2}.
The result below is the analogue of Corollary~\ref{cor:gen-aff} for the 
honeycomb lattice.
\begin{Cor} \label{cor:gen-aff-hc}
Given a solution $(\sigma^+,\sigma^-)$ of the generic honeycomb system
\begin{align} \label{eq:hcl1-g}
\frac{1}{A}\Phi^+ + \frac{1}{B}\Phi^+_{-I, I\!I} +
\frac{1}{S}_{I, II}\Phi^+_{I\!I} & = F^+ \Phi^-,\\
\label{eq:hcl2-g}
\frac{1}{A}\Phi^- + \frac{1}{B_{I,-I\!I}}\Phi^-_{I, -I\!I} +
\frac{1}{S}_{I}\Phi^-_{-I\!I} & = F^- \Phi^+,
\end{align}
then $(\phi^+,\phi^-)$, where $\phi^\pm = \Phi^\pm/\sigma^\pm$, 
satisfies the affine honeycomb linear system \eqref{eq:hcl1}-\eqref{eq:hcl2}
with the coefficients
\begin{equation} \label{eq:coeff-gen-aff-hc}
a = \frac{A}{\sigma^+\sigma^-}, \qquad 
b = \frac{B}{\sigma^+_{-I, I\!I}\sigma^-}, \qquad
s = \frac{S}{\sigma^+_{-I}\sigma^-_{-I,-I\!I}}.
\end{equation}
Conversely, given generic function $(\rho^+, \rho^-):V(\HH)\to\KK$ on 
the honeycomb 
grid, and given a solution $(\phi^+,\phi^-):V(\HH)\to\VV$ 
of the linear problem \eqref{eq:hcl1}-\eqref{eq:hcl2}, then the function 
$(\Phi^+,\Phi^-)=(\phi^+/\rho^+,\phi^-/\rho^-)$ satisfies the system 
\eqref{eq:hcl1-g}-\eqref{eq:hcl2-g}
\eqref{eq:gen-7p-sa}
with the coefficients
\begin{gather*} 
A = \frac{a}{\rho^+\rho^-}, \qquad 
B = \frac{b}{\rho^+_{-I, I\!I}\rho^-}, \qquad
S = \frac{s}{\rho^+_{-I}\rho^-_{-I,-I\!I}}, \\
F^+ = (\rho^-)^2 \left( \frac{1}{a} + \frac{1}{b} + 
\frac{1}{s_{I,I\!I}}\right), \qquad
F^- = (\rho^+)^2 \left( \frac{1}{a} + \frac{1}{b_{I,-I\!I}} + 
\frac{1}{s_{I}}\right).
\end{gather*}
\end{Cor}
\subsection{The duality between the linear problems on
the triangular and the honeycomb lattices}
\label{sec:ST}
As we have already remarked, the honeycomb and triangular tilings are 
dual to each other. 
In fact, the duality can be extended to the
level of the linear problems. Both the honeycomb linear system
\eqref{eq:hcl1}-\eqref{eq:hcl2} and the self-adjoint affine 7-point system
\eqref{eq:s-a-a-7-p} are consequences of the following "duality" equations 
between function $\psi$ on the $\TT$ lattice and the pair
$(\phi^+,\phi^-)$ on the honeycomb lattice
\begin{align} \label{eq:ST1}
\psi_{I} - \psi & = \frac{1}{a} (\phi^+ - \phi^-),\\
\label{eq:ST2}
\psi_{I\!I} - \psi & = \frac{1}{b} (\phi^- - \phi^+_{-I, I\!I}),\\
\label{eq:ST3}
\psi_{I\!I} - \psi_{I} & = \frac{1}{s_{I, I\!I}} (\phi^+_{I\!I} - \phi^-).
\end{align}
\begin{figure}
\begin{center}
\includegraphics[width=4cm]{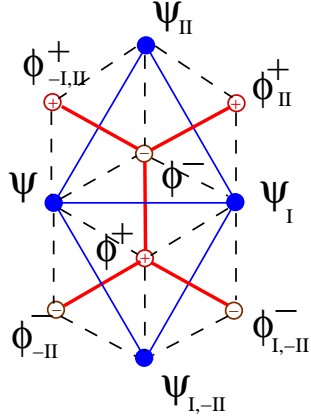}
\end{center}
\caption{The star-triangle duality relations between the triangular and the
honeycomb tilings; the dashed lines represent the edges of the quasiregular
rhombic tiling (the staircase diagonal section)}
\label{fig:S-T-relation}
\end{figure}
Notice that the star-triangle equations \eqref{eq:ST1}-\eqref{eq:ST3} 
are just the Moutard equations on the diagonal staircase section.
Equations \eqref{eq:ST1}-\eqref{eq:ST3} have a point of $\TT_-$ as the star
center, while equations \eqref{eq:ST1}, \eqref{eq:ST2}$_{I,-I\!I}$ and
\eqref{eq:ST3}$_{-I\!I}$ have a point of $\TT_+$ as the star center
(see Figure~\ref{fig:S-T-relation}). 
This is the direct counterpart of the famous
star-triangle relation widely used 
in the analysis of solvable models in statistical mechanics \cite{Baxter}.
As it was shown in \cite{Kashaev} the (nonlinear) star-triangle transformation
in electric networks, interpreted as the local Yang-Baxter equation, leads to
Miwa's discrete BKP equation. The (linear) star-triangle relation 
\eqref{eq:ST1}-\eqref{eq:ST3} gives thus the linear problem for the nonlinear
one (see also Remark \ref{rem:ST-el}).

\begin{Rem}
The above star-triangle relation can be considerd as a particular case of the
relation between the Laplace equations obtained from the discrete
Cauchy--Riemann (discrete Moutard) equation on quad-graphs \cite{BobSur1}, 
where the graph is
the quasiregular rhombic tiling. However, we would like to stress the three
dimensional (not planar) origin of our construction and its geometric meaning.
\end{Rem}

\section{The Laplace transformation of the triangular lattice}
\label{sec:Laplace} 
The Laplace transformations of the self-adjoint 7-point scheme on the triangular
lattice were defined in \cite{Nov} and studied in \cite{NovDyn,DynNov}. The
construction used
in \cite{Nov} consists in the factorization (up to a multiplication operator 
summand) of the second order operator
\eqref{eq:gen-7p-sa} into product of two first order operators. 
In this Section
we will present the geometric meaning of the Laplace transformations. The basic 
idea is as follows: the hyperplanes 
\begin{equation*}
\TT_\ell = \{ (n_1,n_2,n_2)\in\ZZ^3, \; n_1-n_2+n_3 = \ell \}, \qquad \ell \in
\ZZ,
\end{equation*} 
foliate the bipartite $\ZZ^3$ lattice into triangle lattices of a fixed (black
for $\ell$ even, and white for $\ell$ odd) colour. 
\emph{The transition between two
subsequent lattices of the same colour is the Laplace transformation}. 

To develop algebraically this idea we will use formulas 
\eqref{eq:hcl1}-\eqref{eq:hcl2} of the affine honeycomb linear system, which
can be considered as the connection formulas between the
affine self-adjoint 7-point
schemes on $\TT_{\pm 1} = \TT_\pm$ lattices. To simplify our calculations we
will work in the affine gauge induced by the Moutard linear problems (in
\cite{Nov,NovDyn} generic gauge was used). Notice that, in the (forward)
transition from $\TT_-$ to $\TT_+$, we have three
natural possibilities (see Figure~\ref{fig:S-T-relation}) 
\begin{equation*}
\phi^-\to \phi^+, \quad \phi^-\to \phi^+_{I\!I}, \quad
\phi^-\to \phi^+_{-I, I\!I}. 
\end{equation*}
This ambiguity and the corresponding one for the backward
transition was observed on the algebraic level in \cite{NovDyn}.
In this article we choose the transition $\phi^-\to \phi^+$ (the others can be obtained by 
simple
shifts) as the forward Laplace transformation $\mathcal{L}^+$.

The second equation \eqref{eq:hcl2} of the honeycomb linear problem can be
rewritten in the form
\begin{equation} \label{eq:L+}
\phi^+ = \mathcal{L}^+ \phi^- = \frac{1}{r_{I}}\left( b_{I,-I\!I} s_{I} \phi^- +
a  b_{I,-I\!I} \phi^-_{-I\!I} + a s_{I} \phi^-_{I, -I\!I} \right),
\end{equation}
where 
\begin{equation} \label{eq:def-r}
r= a_{-I}s + b_{-I\!I}s + a_{-I} b_{-I\!I}.
\end{equation}
\begin{Rem}
The above transformation is the affine version of the 
Laplace transformation $P_2$ of \cite{NovDyn}. 
\end{Rem}

Inserting such $\phi^+$ into the second equation of 
the linear system \eqref{eq:hcl1}-\eqref{eq:hcl2},
we obtain that the field $\phi^-$ satisfies the affine self-adjoint 7-point
scheme
with the coefficients
\begin{equation} \label{eq:abs-}
a^- = - \frac{a_{I\!I}}{r_{I,I\!I}}, 
\qquad b^- = -\frac{b_{I}}{r_{I,I\!I}},
\qquad s^- = - \frac{s}{r}.
\end{equation}
\begin{Rem}
The coefficients are given up to a common constant multiplier. The above
choice is motivated by the sublattice approach. Indeed, equations
\eqref{eq:phi-pm} and \eqref{eq:subl-notation} imply that, in the sublattice
notation
\begin{equation*} 
a^- = \frac{1}{f^{12}_2}, \qquad 
b^- = \frac{1}{f^{23}_2}, \qquad
s^- = -f^{13}_{-1-3}.
\end{equation*}
We leave it for the Reader to check 
that equations \eqref{eq:abs-} are consequences of the
discrete BKP equations \eqref{eq:nonl-BQL-f}.
\end{Rem}
\begin{Rem} \label{rem:ST-el}
Equations \eqref{eq:abs-} are directly related to the star-triangle
transformations in electrical (or spring) networks, see for example 
\cite{Kashaev}, between the conductivities of conductors (stiffness of the
springs); compare also Section~\ref{sec:ST}.
\end{Rem}
The above equation can be reversed, i.e., the coefficients $a,b,s$ can be
expressed in terms of $a^-, b^-, s^-$ as
\begin{equation} \label{eq:abs-L}
a = - \frac{a^-_{-I\!I}}{q^-_{-I\!I}}, \qquad
b = - \frac{b^-_{-I}}{q^-_{-I}}, \qquad
s = - \frac{s^-}{q^-_{-I,-I\!I}}, \qquad \text{where} \qquad
q^- = a^-b^- + a^-s^-_{I,I\!I} + 
b^-s^-_{I,I\!I}=
\frac{1}{r_{I,I\!I}}.
\end{equation}
This allows to write down the forward Laplace transform of $\phi^-$ in terms of
coefficients \eqref{eq:abs-} of the scheme on $\TT_-$
\begin{equation} \label{eq:Lapl+}
\phi^+ = \mathcal{L}^+ \phi^- = \frac{1}{q^-_{-I\!I}}\left(
b^-_{-I\!I}s^-_{I} \phi^- +
a^-_{-I}  b^-_{-I\!I} \phi^-_{-I\!I} + 
a^-_{-I} s^-_{I} \phi^-_{I, -I\!I} \right).
\end{equation}

Similarly, the backward transformation of $\phi^+$ 
\begin{equation} \label{eq:L-}
\phi^- = \mathcal{L}^{-} \phi^+ = \frac{1}{q}\left( b s_{I, I\!I} \phi^+ +
a  b \phi^+_{I\!I} + a s_{I, I\!I} \phi^+_{-I, I\!I} \right),
\end{equation}
where 
\begin{equation} \label{eq:def-q}
q= a s_{I,I\!I} + b s_{I,I\!I} + a b,
\end{equation}
inserted into the first equation of the linear problem gives 
the affine self-adjoint 7-point scheme
with the coefficients
\begin{equation} \label{eq:tilde-abs}
a^+ = -\frac{a_{I, -I\!I}}{q_{I,-I\!I}}, \qquad 
b^+ = -\frac{b}{q}, \qquad 
s^+ = -\frac{s_{I}}{q_{-I\!I}}. 
\end{equation}
\begin{Rem}
Again, the above coefficients can be obtained by the sublattice approach, i.e., 
\begin{equation*}
a^+ = \frac{1}{f^{12}_1}, \qquad 
b^+ = \frac{1}{f^{23}_1}, \qquad
s^+ = -f^{13}_{-2-3},
\end{equation*}
and equations \eqref{eq:tilde-abs} are consequences of the
nonlinear discrete BKP equations \eqref{eq:nonl-BQL-f}.
\end{Rem}
The inversion formulas 
\begin{equation} \label{eq:abs+L}
a =- \frac{a^+_{-I,I\!I}}{r^+_{I\!I}}, \qquad
b = -\frac{b^+}{r^+_{I\!I}}, \qquad
s = -\frac{s^+_{-I}}{r^+_{-I}}, \qquad \text{where} \qquad
r^+ = a^+_{-I} b^+_{-I\!I} + 
a^+_{-I} s^+ + b^+_{-I\!I} s^+ =
\frac{1}{q_{-I\!I}},
\end{equation}
allow to write down the backward Laplace transform of $\phi^+$ in terms of
coefficients \eqref{eq:tilde-abs} of the scheme on $\TT_+$
\begin{equation} \label{eq:Lapl-}
\phi^- = \mathcal{L}^{-} \phi^+ = \frac{1}{r^+_{I\!I}}\left(
b^+ s^+_{I,I\!I} \phi^+ +
a^+_{-I, I\!I} b^+ \phi^+_{I\!I} + 
a^+_{-I, I\!I} s^+_{I,I\!I} \phi^+_{-I,I\!I} \right).
\end{equation}
We have obtained therefore the sublattice derivation of the Laplace transform,
which we formulate in terms of the proposition below. 
\begin{Prop}
Given a solution $\phi^-$ of the self-adjoint affine 7-point scheme with
coefficients \eqref{eq:abs-}, then its Laplace transform 
$\phi^+=\mathcal{L}^+\phi^-$,
given by formula \eqref{eq:Lapl+}, satisfies a new self-adjoint affine 7-point
scheme on $\TT_+$
whose coefficients can be obtained from those of $\TT^-$ by
inserting expressions \eqref{eq:abs-L} for functions $a,b,s$ into equations 
\eqref{eq:tilde-abs}. The inverse Laplace transform from $\phi^+$ to
$\phi^-=\mathcal{L}^-\phi^+$ is given by equation \eqref{eq:Lapl-}. 
\end{Prop} 

\begin{Cor}
Interpreting the Laplace transformation as a shift in a new discrete variable
$K\in\ZZ$, from equations \eqref{eq:abs-L} and \eqref{eq:abs+L}
one can write down a nonlinear system relating coefficients of the subsequent  
7-point  equations
\begin{equation} \label{eq:deT-1}
\frac{a^K_{-I\!I}}{q^K_{-I\!I}} = 
\frac{a^{K+1}_{-I,I\!I}}{r^{K+1}_{I\!I}}, \qquad
\frac{b^K_{-I}}{q^K_{-I}} = 
\frac{b^{K+1}}{r^{K+1}_{I\!I}}, \qquad
\frac{s^K}{q^K_{-I,-I\!I}} =
\frac{s^{K+1}_{-I}}{r^{K+1}_{-I}},
\end{equation}
where
\begin{equation} \label{eq:deT-2}
q^K = a^Kb^K + a^K s^K_{I,I\!I} + 
b^K s^K_{I,I\!I}, \qquad
r^K = a^K_{-I}b^K_{-I\!I} + 
a^K_{-I} s^K + b^K_{-I\!I} s^K.
\end{equation}
\end{Cor}
\begin{Rem}
As we have seen, the discrete (3D) 
BKP equations \eqref{eq:nonl-BQL-f} can be interpreted as the connection
formulas (or the discrete time evolution)
between subsequent self-adjoint problems on the triangular lattices
$\mathbb{T}_\ell$, $\ell\in\mathbb{Z}$. This implies also that one can obtain
the full Mutard system on the $\ZZ^3$ lattice from the self-adjoint linear
problem on the triangular lattice (or the honeycomb lattice). 
\end{Rem}
\begin{Rem}
The novel, up to our knowledge, nonlinear system 
\eqref{eq:deT-1}-\eqref{eq:deT-2} is the sublattice version
of the discrete BKP equations. It describes the transformation rule of the 
coefficients of self-adjoint 7-point schemes under the Laplace transformation.
The corresponding transformation rule in the case of the discrete affine hyperbolic 
4-point scheme is
equivalent to Hirota's discrete Toda system \cite{Hir}, and was derived in 
\cite{DCN} (see also \cite{Dol-Hir}). Its geometric counterpart is given by the
Laplace transforms of two dimensional discrete conjugate nets, 
introduced in pure
geometric way by Sauer \cite{Sauer2}. The Laplace transform of the 
hyperbolic 4-point
operator (in general gauge) was also defined independently, using the 
factorization 
technique, in \cite{Nov}.
\end{Rem}
We would like to close this Section by presenting the decomposition
of the affine
self-adjoint 7-point scheme in the spirit of the papers 
\cite{Nov,DynNov} (formulas
below are the affine versions of the corresponding formulas in there).
Denote by $T_I$ and $T_{I\! I}$ the standard translation operators along
directions $I$ and $I\! I$ of the triangular lattice. Combining both Laplace
transform \eqref{eq:L+} and \eqref{eq:L-} we obtain equations
\begin{align*}
\phi^+ & = \left[ \frac{1}{r_{I}}\left( b_{I,-I\!I} s_{I}  +
a  b_{I,-I\!I} T_{I\!I}^{-1} + a s_{I} T_{I}T_{I\!I}^{-1} \right)\right]
\left[ \frac{1}{q}\left( b s_{I, I\!I} +
a  b T_{I\!I} + a s_{I, I\!I} T_{I} T_{I\!I}^{-1} \right)
\right]\phi^+, \\
\phi^- & = \left[ \frac{1}{q}\left( b s_{I, I\!I} +
a  b T_{I\!I} + a s_{I, I\!I} T_{I} T_{I\!I}^{-1} \right)
\right]
\left[ \frac{1}{r_{I}}\left( b_{I,-I\!I} s_{I}  +
a  b_{I,-I\!I} T_{I\!I}^{-1} + a s_{I} T_{I}T_{I\!I}^{-1} \right)\right]\phi^-,
\end{align*}
i.e., "factorized" linear problems on $\TT_+$ and 
$\TT_-$. Notice that the fields $a$, $b$, $s$ are the
coefficients of the self-adjoint 7-point scheme \eqref{eq:s-a-a-7-p} of the
intermediate black-point lattice $\TT_0$. 

Elaborating the above formulas, one can derive the following
decomposition of the affine Laplace operator of the lattice $\phi^-$ (we leave
the case of $\phi^+$ to the Reader)
\begin{equation} \label{eq:fact-}
\left[ Q^- (Q^-)^* - w^- \right] \phi^- = 0,
\end{equation}
in terms of the operator
\begin{equation*}
Q^- = x^- + y^- T_{I\!I} + z^- T_{I}^{-1} T_{I\!I},
\end{equation*}
and its formal adjoint
\begin{equation*}
(Q^-)^* = x^- + T_{I\!I}^{-1} y^-  + T_{I} T_{I\!I}^{-1} z^- ,
\end{equation*}
where 
\begin{gather*}
(x^-)^2 = \frac{b^-_{-I\!I}s^-_{I}}{a^-_{-I\!I}}, \qquad
(y^-)^2 = \frac{a^-b^-}{s^-_{I,I\!I}}, \qquad
(z^-)^2 = \frac{a^-_{-I}s^-_{I\!I}}{b^-_{-I}},\\
w^- = \frac{q^-_{-I\! I}}{a^-_{-I\! I}} + 
\frac{q^-_{-I}}{b^-_{-I}} + \frac{q^-}{s^-_{I,I\! I}}.
\end{gather*}
\begin{Rem}
In \cite{NovDyn} the relation between the coefficients \eqref{eq:abs-}
of the 7-point scheme
on $\TT_-$ and the corresponding coefficients of $Q^-$
has been brought to the form 
\begin{gather*}
a^- = y^- z^-_{I}, \qquad
b^- = y^- x^-_{I\! I}, \qquad
s^- = x^-_{-I} z^-_{-I\! I}, \\
w^- = (x^-)^2 + (y^-)^2 + (z^-)^2 + a^- + a^-_{-I} + 
b^- + b^-_{-I\! I} + s^-_{I} + s^-_{I\! I}.
\end{gather*}
\end{Rem}
\begin{Rem}
In the special case $x = y = z= 1$ the operator $Q$ and its adjoint $Q^*$ were
viewed in \cite{DynNov} as the analogs of the operators $\partial$ and
$\bar\partial$ in the discrete holomorphic function theory.
In \cite{Kenyon} the Dirac operator on more general
bipartite planar graphs was defined and related to the dimer model of
statistical mechanics.
\end{Rem}
\begin{Rem}
The Laplace transformation equations \eqref{eq:L+} and \eqref{eq:L-} are
nothing but the linear problem \eqref{eq:hcl1}-\eqref{eq:hcl2}
on the honeycomb lattice. Conversely, the factorization
\eqref{eq:fact-} implies the honeycomb linear
problem.
\end{Rem}

\section{The Darboux transformations} 
\label{sec:DT} 
Our goal here is to derive the Darboux transformation of the self-adjoint
7-point scheme \cite{NSD} and of the linear honeycomb problem 
using the sublattice approach. In doing
that we will follow the corresponding ideas of \cite{DGNS}, where the Darboux
transformation of the self-adjoint 5-point scheme \cite{NSD}
was rederived in such a way.

\subsection{The Darboux transformation of the linear problem on the triangular
lattice}
Before presenting the sublattice derivation of the Darboux transformation,
let us first comment on the geometric content of the calculations below. 
The geometric definition~\cite{BQL} of the B-quadrilateral lattice states 
that there exists
a linear relation between $\psi$, $\psi_{-1-2}$,  $\psi_{-1-3}$ and
$\psi_{-2-3}$; i.e., the corresponding points $[\psi]$, $[\psi_{-1-2}]$,  
$[\psi_{-1-3}]$, $[\psi_{-2-3}]$
of the projectivization of $\VV$ are coplanar. Similarly, the geometric 
definition of the corresponding
reduction of the fundamental transformation states that there exists a linear
relation between $\psi_{-1-3}$, $\psi_{-2-3}$, $\tilde\psi_{-3}$ and
$\tilde\psi_{-1-2-3}$. Both corresponding
planes intersect along the line passing through $[\psi_{-1-2}]$ and 
$[\psi_{-1-3}]$, which means that all the six points belong to a three
dimensional subspace. Therefore any five of the six points
are linearly related (see 
Figure~\ref{fig:Darboux-7p}). This allows, in
particular, for the elimination of $\psi_{-1-3}$, which
does not belong to the allowed subset $V(\TT)$
of the $\ZZ^3$ lattice, yelding a linear relation between 
$\psi$, $\psi_{-1-2}$,
$\psi_{-2-3}$, $\tilde\psi_{-3}$ and $\tilde\psi_{-1-2-3}$, the 
first equation of the Darboux transformation. 
\begin{figure}
\begin{center}
\includegraphics[width=10cm]{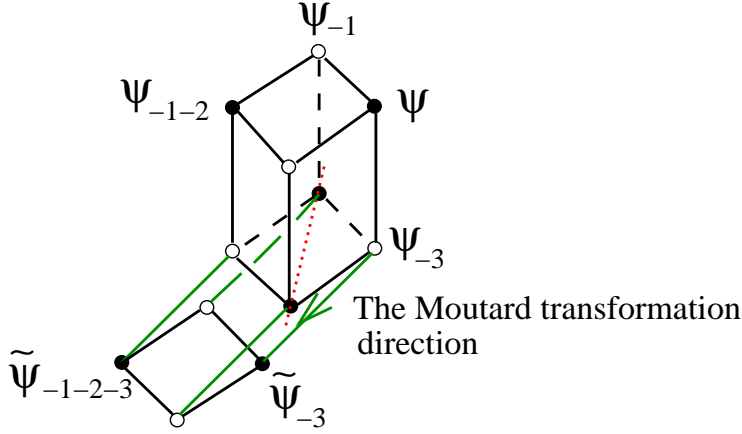}
\end{center}
\caption{Geometry of the sublattice approach to the Darboux transformation of
the self-adjoint 7-point scheme}
\label{fig:Darboux-7p}
\end{figure}

Three discrete Moutard equations \eqref{eq:Moutard-ij}, shifted in $(-i-j)$
directions, lead to the linear equation
\begin{equation} \label{eq:lin-elim1}
\frac{1}{f^{12}_{-1-2}}(\psi-\psi_{-1-2}) -
\frac{1}{f^{13}_{-1-3}}(\psi - \psi_{-1-3}) +
\frac{1}{f^{23}_{-2-3}}(\psi-\psi_{-2-3}) =0.
\end{equation}
If $\theta$ is a scalar
solution of the system \eqref{eq:Moutard-ij}, 
then the Moutard transformations \eqref{eq:Mt}
for $i=1,2$, taken with appropriate shifts, and the discrete Moutard
equation, shifted in the $(-1-2-3))$ direction,
give
\begin{equation}  \label{eq:lin-elim2}
\frac{\theta_{-3}}{\theta_{-2-3}}(\tilde\psi_{-3} -\psi_{-2-3}) -
\frac{\theta_{-3}}{\theta_{-1-3}}(\tilde\psi_{-3} -\psi_{-1-3}) -
\frac{1}{\tilde{f}^{12}_{-1-2-3}}(\tilde\psi_{-3} -\tilde\psi_{-1-2-3})
=0.
\end{equation}
The elimination of $\psi_{-1-3}$ from equations \eqref{eq:lin-elim1} and
\eqref{eq:lin-elim2} leads to
\begin{equation} \label{eq:lin-1-2} 
\tilde\psi_{-3} + \frac{\theta_{-1-3}}{\theta_{-3}}
\frac{1}{\tilde{f}^{12}_{-1-2-3}}(\tilde\psi_{-3} -\tilde\psi_{-1-2-3}) -
\frac{\theta_{-1-3}}{\theta_{-2-3}}(\tilde\psi_{-3} -\psi_{-2-3}) = \\
\psi - \frac{f^{13}_{-1-3}}{f^{12}_{-1-2}}(\psi-\psi_{-1-2}) -
\frac{f^{13}_{-1-3}}{f^{23}_{-2-3}}(\psi-\psi_{-2-3}).
\end{equation}
We multiply both sides by 
$f^{13}_{-1-2-3}\theta_{-2-3}f^{12}_{-1-2}/f^{13}_{-1-3}$,
and we use (a part of) the nonlinear equations \eqref{eq:nonl-BQL-f}:
\begin{equation} \label{eq:ff}
f^{13}_{-1-3}f^{12}_{-1-2-3} = f^{12}_{-1-2}f^{13}_{-1-2-3},
\end{equation}
\begin{equation} \label{eq:BKP-f}
\frac{1}{f^{13}_{-1-2-3}} = f^{23}_{-2-3} + f^{12}_{-1-2} -
\frac{f^{23}_{-2-3} f^{12}_{-1-2}}{f^{13}_{-1-3}},
\end{equation}
to transform equation \eqref{eq:lin-1-2} into
\begin{equation} \label{eq:DT-7-s-1}
\theta_{-3}\tilde\psi_{-3} - \theta_{-1-2-3}\tilde\psi_{-1-2-3} = 
-\frac{\theta_{-2-3}}{f^{23}_{-2-3}}\psi + f^{13}_{-1-2-3} \theta_{-2-3}
\psi_{-1-2} + \left(  \frac{\theta}{f^{23}_{-2-3}} -
f^{13}_{-1-2-3}\theta_{-1-2} \right) \psi_{-2-3}.
\end{equation}

Similar considerations, with equation
\begin{equation}  \label{eq:lin-elim3}
\frac{\theta_{-1}}{\theta_{-1-3}}(\tilde\psi_{-1} -\psi_{-1-3}) -
\frac{\theta_{-1}}{\theta_{-1-2}}(\tilde\psi_{-1} -\psi_{-1-2}) -
\frac{1}{\tilde{f}^{23}_{-1-2-3}}(\tilde\psi_{-1} -\tilde\psi_{-1-2-3})
=0
\end{equation}
instead of equation \eqref{eq:lin-elim2}, lead to  
\begin{equation} \label{eq:DT-7-s-2} 
\theta_{-1}\tilde\psi_{-1} - \theta_{-1-2-3}\tilde\psi_{-1-2-3} = 
\frac{\theta_{-1-2}}{f^{12}_{-1-2}}\psi - f^{13}_{-1-2-3} \theta_{-1-2}
\psi_{-2-3} - \left(  \frac{\theta}{f^{12}_{-1-2}} -
f^{13}_{-1-2-3}\theta_{-2-3} \right) \psi_{-1-2}.
\end{equation}
Introducing the field $\hat\psi = \theta_{-1-2-3}\tilde\psi_{-1-2-3}$ and using 
the
sublattice notation \eqref{eq:subl-notation}, the above system 
\eqref{eq:DT-7-s-1} and \eqref{eq:DT-7-s-2} is rewritten in the form
\begin{align} \label{eq:Darboux-7p1} 
\hat\psi_{I} - \hat\psi \; & = 
- b_{-I\!I}\theta_{-I\!I} \psi - s\theta_{-I\!I} \psi_{-I} +
(b_{-I\!I}\theta + s\theta_{-I}) \psi_{-I\!I}, \\
\label{eq:Darboux-7p2}
\hat\psi_{I\!I} - \hat\psi & = \; \; \; 
a_{-I}\theta_{-I} \psi + s\theta_{-I} \psi_{-II\!} -
(a_{-I}\theta + s\theta_{-I\!I}) \psi_{-I}.
\end{align}
Because $\tilde\psi$ and $\theta^{-1}$ satisfy the same Moutard system, 
then they satisfy the same self-adjoint affine 7-point scheme. By
Corollary~\ref{cor:gen-aff}, the function
$\hat\psi$ satisfies the self-adjoint affine 7-point
scheme whose coefficients can be found from equations \eqref{eq:subl-notation},
with $\tilde{f}^{ij}$ in place of $f^{ij}$, supplemented by the $(-1-2-3)$ shift
and by the change \eqref{eq:coeff-gen-aff}, with respect to the above-mentioned 
gauge, with $\sigma=1/\theta_{-1-2-3}$. This implies that the coefficients read
\begin{align*}
\hat{a} & = \frac{1}{\tilde{f}^{12}_{-1-2-3} \theta_{-3}\theta_{-1-2-3}} =
\frac{1}{f^{12}_{-1-2-3}\theta_{-2-3}\theta_{-1-3}},\\
\hat{b} & = \frac{1}{\tilde{f}^{23}_{-1-2-3} \theta_{-1}\theta_{-1-2-3}} =
\frac{1}{f^{23}_{-1-2-3}\theta_{-1-2}\theta_{-1-3}},\\
\hat{s}_{I I\!I} & = -\tilde{f}^{13}_{-1-3}\frac{1}{\theta_{-1}\theta_{-3}} =
-f^{13}_{-1-3}\frac{1}{\theta\theta_{-1-3}}.
\end{align*}
To rewrite them in the sublattice notation we notice that equation 
\eqref{eq:BKP-f} can be rewritten in the form
\begin{equation} \label{eq:f13-subl}
\frac{1}{f^{13}_{-1-3}} = \frac{r}{s}, 
\qquad 
\end{equation} 
where $r$ is given by \eqref{eq:def-r}.
Then equation \eqref{eq:ff} and its analogue for $f^{23}$ give
\begin{equation} \label{eq:f12-f23-subl}
\frac{1}{f^{12}_{-1-2-3}} = -\frac{a_{-I}}{r}, \qquad 
\frac{1}{f^{23}_{-1-2-3}} = -\frac{b_{-I\!I}}{r}.
\end{equation}
Finally, equation \eqref{eq:lin-elim1}, satisfied by $\theta$, and equation
\eqref{eq:f13-subl} imply
\begin{equation}
\frac{\theta_{-1-3}}{f^{13}_{-1-3}} = 
a_{-I}\theta_{-I} + b_{-I\!I}\theta_{-I\!I} + \frac{\theta a_{-I} b_{-I\!I}}{s},
\end{equation}
which gives
\begin{equation} \label{eq:coeff-DT7}
\hat{a}  = -\frac{a_{-I}}{\theta_{-I\!I} p},\qquad
\hat{b}  = -\frac{b_{-I\!I}}{\theta_{-I} p},\qquad
\hat{s}_{I I\!I}  = -\frac{s}{\theta p},
\end{equation}
where
\begin{equation}\label{eq:p-DT7}
p = s(a_{-I}\theta_{-I} + b_{-I\!I}\theta_{-I\!I}) + a_{-I} b_{-I\!I}\theta .
\end{equation}

Finally we can formulate the theorem, which was given (in the equivalent
non-affine form, see Remark \ref{rem:DT-a-g}  below) in \cite{NSD} directly on 
the 7-point level.
\begin{Th}
Given the general solution $\psi$ of the affine self-adjoint 7-point 
scheme~\eqref{eq:s-a-a-7-p}, and given its particular solution $\theta$, then
equations \eqref{eq:Darboux-7p1} and \eqref{eq:Darboux-7p2} define the general
solution $\hat\psi$ of the affine self-adjoint 7-point scheme with 
coefficients given
by formulas~\eqref{eq:coeff-DT7}-\eqref{eq:p-DT7}.
\end{Th}
\begin{Rem} \label{rem:DT-a-g}
Equations \eqref{eq:Darboux-7p1} and \eqref{eq:Darboux-7p2} and the
transformation formulas \eqref{eq:coeff-DT7} 
for the coefficients coincide 
with equations (32) and (34) in \cite{NSD}. The identification 
$\hat\psi=\psi^\prime \theta$, where $\psi^\prime$ is the corresponding
transform in \cite{NSD}, implies that $\psi^\prime$ 
satisfies a self-adjoint 7-point scheme (an arbitrary gauge does not change the
structure of the equation) which will not have, in general, the affine form.  
\end{Rem}
\begin{Rem}
One could consider another (equivalent) choice of the Darboux transformation 
direction, for
example $\psi\to\theta_{123}\tilde\psi_{123}$. This would change slightly
the Darboux transformation
equations \eqref{eq:Darboux-7p1} and \eqref{eq:Darboux-7p2} and formulas
\eqref{eq:coeff-DT7}-\eqref{eq:p-DT7}.
\end{Rem}
\begin{Rem}
Notice that, for $\psi$ defined on $V(\TT)$, the sum of the
$n_1$, $n_2$ and $n_3$
coordinates of $\tilde\psi_{-1-2-3}$ is odd. However, it is convenient not to 
considered $\tilde\psi_{-1-2-3}$ as a field defined in
points of the white lattice. The proper interpretation would be to consider the
Moutard transformation direction as a shift in the fourth dimension of
the $\ZZ^4$ lattice. Then the sum of all coordinates (including $n_4=1$) of
$\tilde\psi_{-1-2-3}$ is even.
\end{Rem}

\subsection{The Darboux transformation of the integrable honeycomb lattice}
In this last part of the paper we first state the theorem on the Darboux
transformation of the honeycomb linear problem, and then we show its
sublattice derivation.
\begin{Th} \label{th:DT-hc}
Given the general solution $(\phi^+,\phi^-)$ of the honeycomb 
linear system \eqref{eq:hcl1}-\eqref{eq:hcl2},
and given its particular solution $(\theta^+,\theta^-)$,
then the solution $(\check\phi^+,\check\phi^-)$ of the system
\begin{align} 
\label{eq:DT-hc1}
\check\phi^+_{I\!I} - \check\phi^- & = \frac{s}{r}\left( 
\theta^-_{-I\!I}\phi^-_{-I} - \theta^-_{-I}\phi^-_{-I\!I} \right),\\
\label{eq:DT-hc2}
\check\phi^+ - \check\phi^- & = \frac{a_{-I}}{r}\left(
\theta^-_{-I, -I\!I}\phi^-_{-I\!I} -\theta^-_{-I\!I}\phi^-_{-I, -I\!I} \right), \\
\label{eq:DT-hc3}
\check\phi^- - \check\phi^+_{-I, I\!I} & = \frac{b_{-I\!I}}{r}\left(
\theta^-_{-I ,-I\!I}\phi^-_{-I} -\theta^-_{-I}\phi^-_{-I, -I\!I} \right),
\end{align}
with $r$ given by equation \eqref{eq:def-r}, 
satisfies a new honeycomb linear system with the
coefficients 
\begin{equation} \label{eq:coeff-DT-hc}
\check{a} = -\frac{a_{-I}}{r} \theta^{-}_{-I\!I}\theta^{-}_{-I, -I\!I}, \qquad
\check{b} = -\frac{b_{-I\!I}}{r} \theta^{-}_{-I}\theta^{-}_{-I, -I\!I}, \qquad
\check{s}_{I, I\!I} = -\frac{s}{r} \theta^{-}_{-I}\theta^{-}_{-I\!I}.
\end{equation}
\end{Th}
\begin{proof}
The proof can be done by direct verification. However, in the spirit of the
paper, we present its sublattice derivation.
\begin{figure}
\begin{center}
\includegraphics[width=4cm]{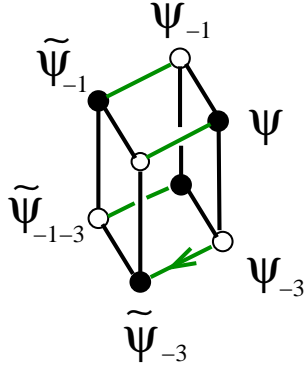}
\end{center}
\caption{Geometry of the sublattice approach to the Darboux transformation of 
the linear problem on the honeycomb grid}
\label{fig:Darboux-hc}
\end{figure}
Consider the wave function $\psi$ defined in the
vertices of the hexahedron, like in
Figure \ref{fig:Darboux-hc}.
The basic geometric property of the B-reduced fundamental transformation (see
Remark~\ref{rem:B-ft}) implies
the possibility of eliminating the "black points of the hexahedron" 
from the discrete Moutard equation
\begin{equation}
\psi - \psi_{-1-3} = f^{13}_{-1-3}(\psi_{-3} - \psi_{-1}),
\end{equation}
and the two Moutard transformation equations
\begin{align}
\tilde\psi - \psi_{-1} & = \frac{\theta_{-1}}{\theta}(\tilde\psi_{-1} - \psi) ,
\\
\tilde\psi_{-3} - \psi_{-1-3} & = \frac{\theta_{-1-3}}{\theta_{-3}}
(\tilde\psi_{-1-3} - \psi_{-3}),
\end{align}
which gives
\begin{equation} \label{eq:DT-hc1-s}
\theta\tilde\psi - \theta_{-1-3}\tilde\psi_{-1-3} = f^{13}_{-1-3} 
(\theta_{-3}\psi_{-1} - \theta_{-1}\psi_{-3}).
\end{equation}
Define $\check\phi^+ =\theta_{-2-3}\tilde\psi_{-2-3}$, 
$\check\phi^- = \theta_{-1-3}\tilde\psi_{-1-3}$
in order to respect the duality transformation, and notice that
$\theta^+ = \theta_{1}$, $\theta^- = \theta_{2}$ are solutions of 
the honeycomb linear
problem \eqref{eq:hcl1}-\eqref{eq:hcl2}. In such notation, and taking into 
consideration equation \eqref{eq:f13-subl}, we can rewrite equation 
\eqref{eq:DT-hc1-s} in the form of the first equation \eqref{eq:DT-hc1} of the
Darboux transformation system.
A similar procedure gives equations
\begin{align*}
\theta_{-2-3}\tilde\psi_{-2-3} - \theta_{-1-3}\tilde\psi_{-1-3} & = 
\frac{1}{f^{12}_{-1-2-3}}(\theta_{-3}\psi_{-1-2-3}-\theta_{-1-2-3}\psi_{-3}), \\
\theta_{-1-3}\tilde\psi_{-1-3} - \theta_{-1-2}\tilde\psi_{-1-2} & = 
\frac{1}{f^{23}_{-1-2-3}}(\theta_{-1}\psi_{-1-2-3}-\theta_{-1-2-3}\psi_{-1}),
\end{align*}
which, using equations  \eqref{eq:f12-f23-subl}, give
the remaining part \eqref{eq:DT-hc2} and \eqref{eq:DT-hc3} of the Darboux
tranformation system.

Because $\tilde\psi$ and $\theta^{-1}$ satisfy the same Moutard system, therefore
$(\tilde\psi_1,\tilde\psi_2)$ and $(\theta^{-1}_1,\theta^{-1}_2)$ satisfy 
the same honeycomb affine linear system. By
Corollary~\ref{cor:gen-aff-hc}, the functions
$(\check\phi^+,\check\phi^-)$ satisfy the honeycomb affine linear system,
the coefficients of which can be found from equations \eqref{eq:subl-notation}
with $\tilde{f}^{ij}$ in place of $f^{ij}$, supplemented by the $(-1-2-3)$ shift
and by the change \eqref{eq:coeff-gen-aff-hc} with respect to the above-mentioned 
gauge with $(\sigma^+,\sigma^-)=(\theta^{-1}_{-2-3},\theta^{-1}_{-1-3})$. 
This implies that the transformed coefficients read
\begin{align*}
\check{a} & = \frac{ \theta_{-2-3}\theta_{-1-3}}{\tilde{f}^{12}_{-1-2-3}} =
\frac{\theta_{-3}\theta_{-1-2-3}}{f^{12}_{-1-2-3}},\\
\check{b} & = \frac{\theta_{-1-2}\theta_{-1-3}}{\tilde{f}^{23}_{-1-2-3} } =
\frac{\theta_{-1}\theta_{-1-2-3}}{f^{23}_{-1-2-3}},\\
\check{s}_{I, I\!I} & = -\tilde{f}^{13}_{-1-3}\theta\theta_{-1-3} =
-f^{13}_{-1-3} \theta_{-1}\theta_{-3},
\end{align*}
which, using equations \eqref{eq:f13-subl}
and \eqref{eq:f12-f23-subl}, take the form as in equation \eqref{eq:coeff-DT-hc}.
\end{proof}

\section{Concluding remarks}
We have derived the integrability properties of the self-adjoint linear problems
on the triangular and honeycomb lattices from the well known theory of the 3D
system of discrete Moutard equations. This result can be the starting point for
several research programs. We would like to point out some of them.

In soliton theory, a linear problem admitting the Darboux type transformations
can be used to construct evolution equations (with the time being discrete or
continuous variable) compatible with the linear system; the standard example is
the KdV hierarchy. As we have already remarked, the Laplace transformation
equations \eqref{eq:deT-1}-\eqref{eq:deT-2} provide such a nonlinear system.
More complicated evolutions of the triangular and honeycomb lattices are 
being studied, see \cite{SDN7p}.

Given integrable reduction of the 3D Moutard system (for example, 
the generalized
isothermic lattice \cite{Dol-isoth}), by restricting it to the quasiregular
rhombic tiling one can obtain the corresponding 
integrable system on the triangular or on the
honeycomb lattice. In particular, to obtain the 
constant angles circle patterns \cite{BobAg} one starts from further reduction
of the generalized isothermic lattice. We remark, that such a process has
already been built in the "consistency around the cube" property 
\cite{NijhoffCC,AdlBobSur}.

Finally, using the "cut and project" method \cite{Senechal}
in conjunction with the sublattice approach \cite{DGNS} one can derive linear
problems (see also Remark \ref{rem:Paolo-gen}) and integrable systems
on more complicated graphs.

\section*{Acknowledgments}
This work was supported by the cultural and
scientific agreement between the University of Roma ``La Sapienza'' 
and the University of Warmia and Mazury in Olsztyn. A.~D.~was also supported 
by
the Polish Ministry of Science and Higher Education grant 1~P03B~017~28. 
M.~N.~was supported by European Community under the Marie Curie Intra-European
Fellowship contract no. MEIF-CT-2005-011228.

\end{document}